\def\UseBibLatex{1}

\documentclass[12pt]{article}

\ifx\sarielComp\undefined%
    \newcommand{\SarielComp}[1]{} \newcommand{\NotSarielComp}[1]{#1}%
\else
    \newcommand{\SarielComp}[1]{#1}%
    \newcommand{\NotSarielComp}[1]{}%
\fi \newcommand{\IfPrinterVer}[2]{#2}%

\usepackage[cm]{fullpage}%
\usepackage{amsmath}%
\usepackage{amssymb}%
\usepackage{dsfont}%
\usepackage{xcolor}%
\usepackage{mathtools}%
\usepackage{commath}

\usepackage{scalerel}%

\SarielComp{\usepackage{sariel_colors}}%

\usepackage[amsmath,thmmarks]{ntheorem}%
\theoremseparator{.}%

\usepackage{titlesec}%
\titlelabel{\thetitle. }%
\usepackage{xcolor}%
\usepackage{mleftright}%
\usepackage{xspace}%
\usepackage{hyperref}%

\IfPrinterVer{%
   \usepackage{hyperref}%
}{%
   \usepackage{hyperref}%
   \hypersetup{%
      breaklinks,%
      ocgcolorlinks, colorlinks=true,%
      urlcolor=[rgb]{0.25,0.0,0.0},%
      linkcolor=[rgb]{0.5,0.0,0.0},%
      citecolor=[rgb]{0,0.2,0.445},%
      filecolor=[rgb]{0,0,0.4},
      anchorcolor=[rgb]={0.0,0.1,0.2}%
   }
}

\providecommand{\BibLatexMode}[1]{}
\providecommand{\BibTexMode}[1]{#1}

\ifx\UseBibLatex\undefined%
  \renewcommand{\BibLatexMode}[1]{}
  \renewcommand{\BibTexMode}[1]{#1}
\else
  \renewcommand{\BibLatexMode}[1]{#1}
  \renewcommand{\BibTexMode}[1]{}
\fi

\BibLatexMode{%
   \usepackage[bibencoding=utf8,style=alphabetic,backend=biber]{biblatex}%
   \usepackage{sariel_biblatex}%
}

\theoremseparator{.}%

\theoremstyle{plain}%
\newtheorem{theorem}{Theorem}[section]

\newtheorem{lemma}[theorem]{Lemma}

\newtheorem{corollary}[theorem]{Corollary}
\newtheorem{claim}[theorem]{Claim}%

\theoremstyle{plain}%
\theoremheaderfont{\sf} \theorembodyfont{\upshape}%
\newtheorem*{remark:unnumbered}[theorem]{Remark}%

\newcommand{\myqedsymbol}{\rule{2mm}{2mm}}

\theoremheaderfont{\em}%
\theorembodyfont{\upshape}%
\theoremstyle{nonumberplain}%
\theoremseparator{}%
\theoremsymbol{\myqedsymbol}%
\newtheorem{proof}{Proof:}%

\definecolor{blue25emph}{rgb}{0, 0, 11}
\providecommand{\emphic}[2]{%
   \textcolor{blue25emph}{%
      \textbf{\emph{#1}}}%
   \index{#2}}

\providecommand{\emphi}[1]{\emphic{#1}{#1}}

\definecolor{almostblack}{rgb}{0, 0, 0.3}

\newcommand{\atgen}{\symbol{'100}}
\newcommand{\EliotThanks}[1]{%
   \thanks{%
      Department of Computer Science; University of Illinois; 201
      N. Goodwin Avenue; Urbana, IL, 61801, USA; {\tt
         erobson2\atgen{}illinois.edu}; {\tt
         \url{https://eliotwrobson.github.io/}.} #1}}
\newcommand{\SarielThanks}[1]{\thanks{Department of Computer Science;
      University of Illinois; 201 N. Goodwin Avenue; Urbana, IL,
      61801, USA; {\tt sariel\atgen{}illinois.edu}; {\tt
         \url{http://sarielhp.org/}.} #1}}

\newcommand{\HLink}[2]{\hyperref[#2]{#1~\ref*{#2}}}
\newcommand{\HLinkSuffix}[3]{\hyperref[#2]{#1\ref*{#2}{#3}}}

\newcommand{\corlab}[1]{\label{cor:#1}}

\newcommand{\clmlab}[1]{\label{claim:#1}}
\newcommand{\clmref}[1]{\HLink{Claim}{claim:#1}}

\newcommand{\lemlab}[1]{\label{lemma:#1}}
\renewcommand{\lemref}[1]{\HLink{Lemma}{lemma:#1}}%

\providecommand{\eqlab}[1]{}%
\renewcommand{\eqlab}[1]{\label{equation:#1}}

\providecommand{\remove}[1]{}%
\newcommand{\Set}[2]{\left\{ #1 \;\middle\vert\; #2 \right\}}
\newcommand{\pth}[2][\!]{\mleft({#2}\mright)}%

\renewcommand{\th}{th\xspace}

\renewcommand{\Re}{\mathbb{R}}%

\usepackage[inline]{enumitem}

\newlist{compactenumA}{enumerate}{5}%
\setlist[compactenumA]{topsep=0pt,itemsep=-1ex,partopsep=1ex,parsep=1ex,%
   label=(\Alph*)}%

\newlist{compactenuma}{enumerate}{5}%
\setlist[compactenuma]{topsep=0pt,itemsep=-1ex,partopsep=1ex,parsep=1ex,%
   label=(\alph*)}%

\newlist{compactenumI}{enumerate}{5}%
\setlist[compactenumI]{topsep=0pt,itemsep=-1ex,partopsep=1ex,parsep=1ex,%
   label=(\Roman*)}%

\newlist{compactenumi}{enumerate}{5}%
\setlist[compactenumi]{topsep=0pt,itemsep=-1ex,partopsep=1ex,parsep=1ex,%
   label=(\roman*)}%

\newlist{compactitem}{itemize}{5}%
\setlist[compactitem]{topsep=0pt,itemsep=-1ex,partopsep=1ex,parsep=1ex,%
   label=\bullet}%

\DeclareFontFamily{U}{BOONDOX-calo}{\skewchar\font=45 }
\DeclareFontShape{U}{BOONDOX-calo}{m}{n}{
  <-> s*[1.05] BOONDOX-r-calo}{}
\DeclareFontShape{U}{BOONDOX-calo}{b}{n}{
  <-> s*[1.05] BOONDOX-b-calo}{}
\DeclareMathAlphabet{\mathcalb}{U}{BOONDOX-calo}{m}{n}
\SetMathAlphabet{\mathcalb}{bold}{U}{BOONDOX-calo}{b}{n}
\DeclareMathAlphabet{\mathbcalb}{U}{BOONDOX-calo}{b}{n}

\usepackage{stmaryrd}%

\providecommand{\Mh}[1]{#1}%

\newcommand{\q}{\Mh{q}}%
\newcommand{\p}{\Mh{p}}%
\renewcommand{\P}{\Mh{P}}%
\renewcommand{\L}{\Mh{L}}%
\newcommand{\Q}{\Mh{Q}}%

\newcommand{\cenX}[1]{\Mh{\overline{\mathsf{c}}}_{#1}}
\newcommand{\permut}[1]{\left\langle {#1} \right\rangle}
\newcommand{\DotProdY}[2]{\permut{{#1},{#2}}}

\newcommand{\dY}[2]{\left\| {#1} {#2} \right\|}%
\newcommand{\dwY}[2]{\left\| \smash{{#1} {#2}} \right\|}%

\numberwithin{figure}{section}%
\numberwithin{table}{section}%
\numberwithin{equation}{section}%

\newcommand{\ProjWidthC}{\Mh{\overline{\omega}}}%

\newcommand{\pwY}[2]{\ProjWidthC\pth{#1, #2}}

\newcommand{\avgX}[1]{\Mh{\mu}({#1})}
\newcommand{\energyX}[1]{\Mh{\sigma}\pth{#1}}%
\newcommand{\egX}[1]{\energyX{#1}}%
\newcommand{\One}{\mathds{1}}%

\newcommand{\Sn}{\Mh{\Delta}_n}
\newcommand{\Dn}{\Mh{D}_n}
\newcommand{\ee}{\Mh{\mathcalb{e}}}
\newcommand{\hp}{\Mh{\mathcalb{h}}}
\newcommand{\Hc}{\Mh{\mathcal{H}}}

\BibLatexMode{%
   \bibliography{width_simplex} %
}

\begin{document}

\title{On the Width of the Regular $n$-Simplex}

\author{Sariel Har-Peled\SarielThanks{Work on this paper
      was partially supported by a NSF AF award
      CCF-1907400.
   }
   \and
   Eliot W. Robson%
   \EliotThanks{}%
}

\date{\today}

\maketitle

\begin{abstract}
    Consider the regular $n$-simplex $\Sn$ -- it is formed by the
    convex-hull of $n+1$ points in Euclidean space, with each pair of
    points being in distance exactly one from each other. We prove an
    exact bound on the width of $\Sn$ which is $\approx \sqrt{2/n}$.
    Specifically,
    \begin{math}
        \mathrm{width}(\Sn) = \sqrt{\frac{2}{n + 1}}
    \end{math}
    if $n$ is odd, and
    \begin{math}
        \mathrm{width}(\Sn) = \sqrt{\frac{2(n+1)}{n(n+2)}}
    \end{math}
    if $n$ is even.  While this bound is well known
    \cite{gk-iojrc-92,a-wds-77}, we provide a self-contained
    elementary proof that might (or might not) be of interest.
\end{abstract}

\section{The width of the regular simplex}

A regular $n$-simplex $\Sn$ is a set of $n + 1$ points in Euclidean
space such that every pair of points is in distance exactly $1$ from
each other. For simplicity, it is easier to work with the
simplex $\Dn$ formed by the convex-hull of
$\ee_1,\ldots, \ee_{n+1} \in \Re^{n+1}$, where $\ee_i$ is the $i$\th
standard unit vector\footnote{That is, $\ee_i$ is $0$ in all
   coordinates except the $i$\th coordinate where it is $1$.}.
Observe that $\bigl.\dwY{\smash{\ee_i}}{\smash{\ee_j}} = \sqrt{2}$,
which implies that $\Dn = \sqrt{2}\Sn$.  All the vertices of $\Dn$
lie on the hyperplane
$\hp \equiv \sum_{i=1}^{n+1} x_i = \DotProdY{(x_1,\ldots,
   x_{n+1})}{\One}= 1$, where $\One = (1, 1, \ldots , 1)$.  In
particular, the point $\cenX{n} = \One/(n + 1) \in \Dn$ is the center
of $\Dn$, and is in equal distance
\begin{equation*}
    \widehat{R_n}
    =
    \dY{\cenX{n}}{\ee_i}
    =
    \sqrt{n \frac{1}{(n+1)^2} + \pth{1-\frac{1}{n+1}}^2}
    =
    \sqrt{1 - \frac{2}{n+1} + \frac{n+1}{(n+1)^2}}
    =
    \sqrt{\frac{n}{n+1}}.
\end{equation*}
from all the vertices of $\Dn$. Note that the largest ball one can
place in $\hp$, which is still contained in $\Dn$, is centered at
$\cenX{n}$. Furthermore, it has radius
\begin{equation*}
    \widehat{r_n}
    =%
    \dY{\cenX{n}}{-(0,\cenX{n-1})}
    =%
    \sqrt{\frac{1}{(n+1)^2} + n \pth{\frac{1}{n} - \frac{1}{n+1}}^2}
    =%
    \sqrt{\frac{n+1}{n(n+1)^2}}
    =%
    \frac{1}{\sqrt{n(n+1)}}.
\end{equation*}

For a unit vector $u$, and a set $\P \subseteq \Re^{n+1}$, let
\begin{equation*}
    \pwY{u}{\P}
    =
    \max_{p \in \P} \DotProdY{u}{p} - \min_{p \in \P}
    \DotProdY{u}{p}
\end{equation*}
denote the \emphi{projection width} of $\P$ in the direction of $u$.
The \emphi{width} of $\P$ is the minimum projection
width over all directions.

\paragraph{Energy of points.}

For a vector $v = (v_1 , \ldots, v_{n+1}) \in \Re^{n+1}$, let
\begin{math}
    \avgX{v} = \sum_{i} v_i / (n+1),
\end{math}
and let
\begin{equation*}
    \widehat{v} = v - \avgX{v} \One,
\end{equation*}
be the translation of $v$ by its centroid $\avgX{v} \One$ so that
$\DotProdY{\widehat{v}}{\One} = 0$.  The \emphi{energy} of $v$ is
$\egX{v} = \norm{\widehat{v}}^2$.  The energy is the minimum $1$-mean
clustering price of the numbers $v_1 ,\ldots, v_{n+1}$. We need the
following standard technical claim, which implies that if we move a
value away from the centroid, the $1$-mean clustering price of the set
goes up.

\begin{claim}
    \clmlab{energy}%
    Let $v = (v_1 , \ldots, v_{n+1}) \in \Re^{n+1}$ be a point, and
    let $u$ be a point that is identical to $v$ in all coordinates
    except the $i$\th one, where $u_i > v_i \geq \avgX{v}$. Then
    $\egX{u} > \egX{v}$. The same holds if $u_i < v_i < \avgX{v}$.
\end{claim}

\begin{proof}
    Let $\delta_i = (u_i - v_i )/(n + 1)$, and observe that
    $\avgX{u} = \avgX{v} + \delta_i$. As such, we have
    \begin{equation*}
        \egX{u}
        =%
        \sum\nolimits_j (u_j - \avgX{u})^2
        =%
        \sum\nolimits_j (v_j - \avgX{v}+\delta_i)^2
        - (v_i - \avgX{v} + \delta_i)^2 +
        (u_i - \avgX{v} + \delta_i)^2.
    \end{equation*}
    Since
    \begin{math}
        \sum\nolimits_j (v_j - \avgX{v}) = 0,
    \end{math}
    we have that
    \begin{equation*}
        \sum\nolimits_j (v_j -\avgX{v} +\delta_i)^2
        =%
        \sum\nolimits_j (v_j -\avgX{v})^2
        +
        \pth{2\delta_i \sum\nolimits_j (v_j -\avgX{v})}
        +
        \sum\nolimits_j \delta_i^2
        =
        \egX{v}
        + (n+1) \delta_i^2.
    \end{equation*}
    Rearranging the above and using that $u_i > v_i\geq \avgX{v}$, we have
    \begin{equation*}
        \egX{u} - \egX{v}
        =%
        (n+1) \delta_i^2
        + (u_i - \avgX{v} + \delta_i)^2
        - (v_i - \avgX{v} + \delta_i)^2
        >
        (u_i - v_i)
        (u_i + v_i -2\avgX{v}+ 2\delta_i) >0.
    \end{equation*}
\end{proof}

\begin{lemma}
    \lemlab{w:odd}%
    For $n$ odd, the width of $\Dn$ is $2/\sqrt{n+1}$, and this is
    realized by they projection width of all the directions in
    $\Hc = \Set{\smash{v/ \sqrt{n + 1}}}{ v \in \{-1, +1\}^{n+1}
       \text{ and } \DotProdY{v}{\One} =0}\Bigr.$ (and no other
    direction).
\end{lemma}

\begin{proof}
    Consider a unit vector $z$ that realizes the minimum width of
    $\Dn$ – here, in addition to $\norm{z} = 1$, we also require that
    $\DotProdY{z}{\One} = 0$, as one has to consider only directions
    that are parallel to the hyperplane containing $\Dn$.  To this
    end, let
    \begin{math}
        \beta = \max_i z_i
    \end{math}
    and
    \begin{math}
        \alpha = \min_i z_i
    \end{math}
    and observe that
    \begin{equation*}
        \mathrm{width}(\Dn)
        =%
        \pwY{z}{\Dn}
        =%
        \max_{i} \DotProdY{z}{\ee_i} - \min_{i}
        \DotProdY{z}{\ee_i}
        =
        \beta - \alpha.
    \end{equation*}
    Next, Consider the point $u$, where for all $i$ we set
    \begin{equation*}
        u_i
        =
        \begin{cases}
          \alpha  & z_i < 0\\
          \beta & z_i \geq 0.
        \end{cases}
    \end{equation*}
    A careful repeated application of \clmref{energy}, implies that
    $\egX{u} >\egX{z}$ if any coordinate of $z$ is not already either
    $\alpha$ or $\beta$. But then the point $\widehat{u}$ has (i)
    ``width'' $\beta - \alpha$, (ii) $\norm{\widehat{u}} > 1$, and
    (iii) $\DotProdY{\widehat{u}}{\One}=0$. But this implies that the
    projection width of $\Dn$ on $\widehat{u} / \norm{\widehat{u}}$ is
    $(\beta-\alpha)/ \norm{\widehat{u}} < \beta - \alpha$, which is a
    contradiction to the choice of $z$.

    Thus, it must be that all the coordinates of $z$ are either
    $\alpha$ or $\beta$. Let $t$ be the number of coordinates of $z$
    that are $\alpha$, and observe that
    \begin{equation*}
        \egX{z}
        =%
        \norm{z}^2
        =%
        t\alpha^2  + (n+1-t)\beta^2
        =
        1
        \qquad\text{and}\qquad%
        \DotProdY{z}{\One}= t\alpha + (n+1-t) \beta = 0.
    \end{equation*}
    This implies that $\beta = - \frac{t}{n+1-t}\alpha$, and thus
    \begin{equation*}
        t\alpha^2 + \frac{(n+1-t)t^2}{(n+1-t)^2} \alpha^2 = 1
        \quad\implies\quad%
        \frac{t(n+1-t)+ t^2}{n +1 -t} \alpha^2 = 1
        \quad\implies\quad%
        \alpha = -
        \sqrt{\frac{n+1-t}{t(n+1)}}.
    \end{equation*}
    Thus, the width of $\Dn$ is
    \begin{math}
        \beta - \alpha =%
        \bigl(1 + \tfrac{t}{n+1-t} \bigr) \sqrt{\frac{n+1-t}{t(n+1)}}
        =%
        \sqrt{\frac{n+1}{t(n+1-t)}}.
    \end{math}
    The last quantity is minimized when the denominator is maximized,
    which happens for $t = (n+1)/2$.  Namely, the width of $\Dn$ is
    $ 2/\sqrt{n+1}$.

    We have that $t\alpha + t\beta = 0$, which implies that
    $\alpha = -\beta$ and thus $\beta = 1/\sqrt{n + 1}$. It follows
    that $z \in \Hc$.
\end{proof}

\begin{lemma}
    \lemlab{w:even}%
    For $n$ even, the width of $\Dn$ is
    $2 \sqrt{ \smash{\frac{n+1}{n(n+2)}}\bigr.}$.
\end{lemma}

\begin{proof}
    The proof of \lemref{w:odd} goes through with minor
    modifications. The minimum value is realized by $t = n/2$. This
    implies that
    \begin{equation*}
        \alpha
        =%
        - \sqrt{\frac{n+1-t}{t(n+1)}}
        =%
        - \sqrt{\frac{n+2}{n(n+1)}}
        \qquad\text{and}\qquad%
        \beta
        =%
        -\frac{t}{n+1-t}\alpha
        =%
        \sqrt{\frac{n}{(n+1)(n+2)}}.
    \end{equation*}
    As a sanity
    check, observe that
    \begin{math}
        t\alpha^2 + (n + 1 - t)\beta^2%
        =%
        \frac{n}{2}\frac{n+2}{n(n+1)} + \frac{n+2}{2}
        \frac{n}{(n+1)(n+2)}%
        =%
        1.
    \end{math}
    Thus, the width of
    $\Dn$ is
    \begin{equation*}
        \beta - \alpha
        =
        \sqrt{\frac{n}{(n+1)(n+2)}}
        +
        \sqrt{\frac{n+2}{n(n+1)}}
        =%
        \frac{2(n+1)}{\sqrt{n(n+1)(n+2)}}
        =%
        2 \sqrt{\frac{n+1}{n(n+2)}}.
    \end{equation*}
\end{proof}

Using that $\Sn = \Dn/ \sqrt{2}$ and rescaling the above bounds, we
get the following.

\begin{corollary}
    \corlab{bounds}%
    The width of the regular $n$-simplex $\Sn$ is
    \begin{math}
        \sqrt{2/(n + 1)}
    \end{math}
    if $n$ is odd. The width is
    \begin{math}
        \sqrt{\frac{2(n+1)}{n(n+2)}}
    \end{math}
    if $n$ is even.  The inradius (i.e., radius of largest ball inside
    $\Sn$) is
    \begin{math}
        r_n = 1/\sqrt{2n(n+1)},
    \end{math}
    and the circumradius (i.e.,
    radius of minimum ball enclosing $\Sn$ is
    \begin{math}
        R_n = \sqrt{\frac{n}{2(n+1)}}.
    \end{math}
\end{corollary}

\BibTexMode{%
   \bibliographystyle{alpha}
   \bibliography{width_simplex}
}%
\BibLatexMode{\printbibliography}

\end{document}